\documentclass[runningheads,a4paper]{llncs}

\usepackage{amsmath, amssymb, enumitem, centernot, url, tensor, xcolor}
\usepackage{bbm}
\usepackage{subfig}
\usepackage{tikz}

\newtheorem{myclaim}{Claim}

\newcommand{\N}{\mathbb{N}}

\newcommand{\wH}{w_\mathrm{H}}

\newcommand{\X}{\mathrm{X}}
\newcommand{\id}{\mathrm{id}}

\newcommand{\lcm}{\mathrm{lcm}}
\newcommand{\notimplies}{\centernot\implies}

\newcommand{\keywords}[1]{\par\addvspace\baselineskip
\noindent\keywordname\enspace\ignorespaces#1}

\begin{document}
\mainmatter  

\title{Commutative automata networks}

\titlerunning{Commutative automata networks}

\author{Florian Bridoux\inst{1} \and Maximilien Gadouleau\inst{2} \and Guillaume Theyssier\inst{3}}

\authorrunning{Commutative automata networks}

\institute{Universit\'e d'Aix-Marseille, CNRS, Centrale Marseille, LIS, Marseille, France. florian.bridoux@lis-lab.fr.
\and
Department of Computer Science, Durham University, South Road, Durham, DH1 3LE, UK. m.r.gadouleau@durham.ac.uk.
\and
Universit\'e d'Aix-Marseille, CNRS, Centrale Marseille, I2M, Marseille, France. guillaume.theyssier@cnrs.fr.
}

\toctitle{Commutative automata networks}
\tocauthor{Florian Bridoux, Maximilien Gadouleau and Guillaume Theyssier}
\maketitle



\begin{abstract}
Automata networks are mappings of the form $f : Q^Z \to Q^Z$, where $Q$ is a finite alphabet and $Z$ is a set of entities; they generalise Cellular Automata and Boolean networks. An update schedule dictates when each entity updates its state according to its local function $f_i : Q^Z \to Q$. One major question is to study the behaviour of a given automata networks under different update schedules. In this paper, we study automata networks that are invariant under many different update schedules. This gives rise to two definitions, locally commutative and globally commutative networks. We investigate the relation between commutativity and different forms of locality of update functions; one main conclusion is that globally commutative networks have strong dynamical properties, while locally commutative networks are much less constrained. We also give a complete classification of all globally commutative Boolean networks.

\end{abstract}

\section{Introduction} \label{sec:introduction}

Automata networks are mappings of the form $f : Q^Z \to Q^Z$, where $Q$ is a finite alphabet and $Z$ is a set of entities; they generalise Cellular Automata and Boolean networks. An update schedule dictates when each entity updates its state according to its local function $f_i : Q^Z \to Q$. 

One major question is to study the behaviour of a given automata network under different update schedules. A lot is known about the relation between synchronous and asynchronous dynamics of finite or infinite interaction networks (see \cite{GN12,NS17} and references therein). In particular, there is a stream of work that focuses on networks with the most ``asynchronous power,'' i.e. networks which can asynchronously simulate a large amount of parallel dynamics \cite{CFG14,BCG19}. 

In this paper, we study automata networks that are ``robust to asynchronicity,'' i.e. invariant under many different update schedules. 
This gives rise to two definitions, locally commutative and globally commutative networks. A network is globally commutative if applying $f$ in parallel is equivalent to updating different parts of $Z$ sequentially, for any possible ordered partition of $Z$. A network is locally commutative if applying $f$ in parallel is equivalent to updating different elements of $Z$ sequentially, for any possible ordered partition of $Z$ into singletons. (Formal definitions will be given in Section \ref{sec:commutative}.)

In this paper, we investigate the relation between commutativity (defined in section~\ref{sec:commutative}) and different forms of locality of update functions, namely: finiteness properties of the influences of some nodes on some other ones (section~\ref{sec:finiteness}), dynamical locality (section~\ref{sec:dynamic_locality}), and idempotence (section~\ref{sec:idempotence}) ; one main conclusion is that globally commutative networks have strong dynamical properties, while locally commutative networks are much less constrained. We also determine the transition graphs of all globally commutative Boolean networks (section~\ref{sec:classification}).

\section{Commutativity properties} \label{sec:commutative}

Let $Q$ be a finite alphabet of size $q \ge 2$, and let $Z$ be a countable set. We refer to any element $x \in Q^Z$ as a \textbf{configuration}. An automata network, or simply \textbf{network}, is any mapping $f : Q^Z \to Q^Z$. 

We denote the set of all subsets of $Z$ as $\mathcal{P}(Z)$ and the set of all finite subsets of $Z$ as $\mathcal{FP}(Z)$. For all $j \in \N$, we denote $[j] = \{1, \dots, j\}$. For any $x \in Q^Z$ and any $s \in \mathcal{P}(Z)$, we will use the shorthand notation $x_s = (x_i : i \in s)$ and $x = (x_s, x_{Z \setminus s})$ (the ordering of the elements of $s$ will be immaterial). We will extend these notations to networks as well, i.e. $f = (f_i : i \in Z)$ and $f_s = (f_i: i \in s)$. For any $a \in Q$ and $s \in \mathcal{P}(Z)$, we denote the configuration $x \in Q^s$ with $x_i = a$ for all $i \in s$ as $a_s$ (we shall commonly use examples such as $0_Z$ and $1_Z$). If $Q = \{0,1\}$ and $Z = \N$, then the \textbf{density} of $x \in \{0,1\}^\N$ is 
\[
	\delta(x) := \limsup_{n \to \infty} \frac{\wH(x_1, \dots, x_n)}{n},
\]
where the Hamming weight $\wH(x_1, \dots, x_n)$ is the number of nonzero elements of $(x_1, \dots, x_n)$.

For any $s \in \mathcal{P}(Z)$, the update of $s$ according to $f$ is $f^{(s)} : Q^Z \to Q^Z$ such that $f^{(s)}(x) = (f_s(x), x_{Z \setminus s})$. We denote $f^{(s_1, \dots, s_k)} = f^{(s_k)} \circ \dots \circ f^{(s_1)}$ for any finite sequence $s_1, \dots, s_k \in \mathcal{P}(Z)$.
We consider the following commutativity properties.
\begin{enumerate}[label=\textbf{(C\arabic*)}]
	\item \label{item:Ci}
	$f^{(i,j)} = f^{(j,i)}$ for all $i,j \in Z$.

	\item \label{item:Cb}
	$f^{(b,c)} = f^{(c,b)}$ for all $b,c \in \mathcal{FP}(Z)$.

	\item \label{item:Cs}
	$f^{(s,t)} = f^{(t,s)}$ for all $s,t \in \mathcal{P}(Z)$.
\end{enumerate}

We say that a network satisfying \ref{item:Ci} is \textbf{locally commutative}, while a network satisfying \ref{item:Cs} is \textbf{globally commutative}. We remark that if $f$ is locally (respectively, globally) commutative, then $f^{(s)}$ is locally (respectively, globally) commutative for any $s \in \mathcal{P}(Z)$.

Throughout this paper, we use notation ${\bf (X)} \implies {\bf (Y)}$ to mean that if a network satisfies Property ${\bf (X)}$, then it satisfies Property ${\bf (Y)}$ as well. We use shorthand notation such as ${\bf (X)} \implies {\bf (Y)} \implies {\bf (Z)}$ with the obvious meaning, and we use the notation ${\bf (X)} \notimplies {\bf (Y)}$ to mean that there exists a network satisfying ${\bf (X)}$ but not ${\bf (Y)}$. For instance, we have $\ref{item:Cs} \implies \ref{item:Cb} \implies \ref{item:Ci}$.

The commutativity properties above have alternative definitions. An \textbf{enumeration} of $Z$ is an ordered partition of $Z$, i.e. a sequence $Y = (y_\tau : \tau \in \N)$, where $ \bigcup_{\tau \in \N} y_\tau = Z$ and $y_\tau \cap y_\sigma = \emptyset$ for all $\tau \ne \sigma$. If each $y_\tau$ is either a singleton or the empty set, then $Y$ is a sequential enumeration; if each $y_\tau$ is finite, then $Y$ is a finite-block enumeration. For any enumeration $Y$ and any $t \in \N$, we denote $Y_t = (y_1, \dots, y_t)$. We then define $f^Y : Q^Z \to Q^Z$ by
	${f^Y_{y_\tau} := f^{Y_\tau}_{y_\tau}}$ for all ${\tau \in \N}$.
Since all the $y_\tau$ are disjoint, we see that $f^{Y_t}_{y_\tau} = f^{Y_\tau}_{y_\tau} = f^Y_{y_\tau}$ for any $t \ge \tau$. We then consider the following alternative commutativity properties.
\begin{enumerate}[label=\textbf{(C\arabic*a)}]
	\item \label{item:Cia}
	$f^I = f^J$ for any two sequential enumerations $I$ and $J$ of $Z$.

	\item \label{item:Cba}
	$f^B = f^C$ for any two finite-block enumerations $B$ and $C$ of $Z$.

	\item \label{item:Csa}
	$f^S = f^T$ for any two enumerations $S$ and $T$ of $Z$.
\end{enumerate}
We also consider the following alternative properties, which at first sight are stronger than those listed just above.
\begin{enumerate}[label=\textbf{(C\arabic*b)}]
	\item \label{item:Cib}
	$f = f^I$ for any sequential enumeration $I$ of $Z$.

	\item \label{item:Cbb}
	$f = f^B$ for any finite-block enumeration $B$ of $Z$.

	\item \label{item:Csb}
	$f = f^S$ for any enumeration $S$ of $Z$.
\end{enumerate}

We prove that in fact, those three versions are equivalent. This helps us show that \ref{item:Ci} and \ref{item:Cb} are also equivalent. 

%
%
%

\begin{theorem} \label{th:equivalence_commutative}
Commutativity properties are related as follows:
\begin{align}
	\label{item:equivalence_singleton}
	\ref{item:Ci} &\iff \ref{item:Cia} \iff \ref{item:Cib}.	\\
	\label{item:equivalence_finite}
	\ref{item:Cb} &\iff \ref{item:Cba} \iff \ref{item:Cbb}.	\\
	\label{item:equivalence_all}
	\ref{item:Cs} &\iff \ref{item:Csa} \iff \ref{item:Csb}.	\\
        \label{item:CiCb}
	\ref{item:Ci} &\iff \ref{item:Cb} \notimplies \ref{item:Cs}.
\end{align}
\end{theorem}

\begin{proof}
\eqref{item:equivalence_singleton}-\eqref{item:equivalence_all}. We only prove the equivalence \eqref{item:equivalence_all} for \ref{item:Cs}; the other equivalences are similarly proved. Firstly, $\ref{item:Csb} \implies \ref{item:Csa}$. Secondly, if $f$ satisfies \ref{item:Csa}, then 
\[
	f^{(s,t)}_t = f^{(s,t, Z \setminus (s \cup t) )}_t = f^{(t,s, Z \setminus (s \cup t) )}_t = f^{(t,s)}_t,
\]
and by symmetry $f^{(s,t)}_s = f^{(t,s)}_s$, thus $f^{(s,t)} = f^{(t,s)}$ and $f$ satisfies \ref{item:Cs}. Thirdly, if $f^{(s,t)} = f^{(t,s)}$ for all $s,t$, then let $S = (s_\tau : \tau \in \N)$; we have 
\[
	f^S_{s_\tau} = f^{S_\tau}_{s_\tau} = f^{(s_\tau , s_1, \dots, s_{\tau-1})}_{s_\tau} = f_{s_\tau}
\]
for any $\tau$ and hence $f^S = f$.

\eqref{item:CiCb}. We first prove $\ref{item:Ci} \implies \ref{item:Cb}$. If $f$ satisfies \ref{item:Ci}, then for any finite subsets $b = \{ b_1, \dots, b_k \}$ and $c = \{ c_1, \dots, c_l \}$, the equivalence in \eqref{item:equivalence_singleton} yields
\[
	f^{(b,c)} = f^{(b_1, \dots, b_k, c_1, \dots, c_l)} = f^{(c_1, \dots, c_l, b_1, \dots, b_k)} = f^{(c,b)}.
\]

We now prove $\ref{item:Ci} \notimplies \ref{item:Cs}$ by exhibiting a network satisfying \ref{item:Ci} but not \ref{item:Cs}. Let $Q = \{0,1\}$ and $Z = \N$, then let $f(x) = 1_Z$ if $\delta(x) \le 1/2$ and $f(x) = 0_Z$ otherwise. Since $\delta( f^{(i)}(x)) = \delta(x)$ for any $i \in Z$, we easily obtain that $f^{(i)}$ and $f^{(j)}$ commute. On the other hand, let $x = 0_Z$ and $s$ a set of density $1/3$ (e.g. all multiples of $3$). Then $f^{(s, Z \setminus s)}_s (x) = 1_s$ while $f^{(Z \setminus s, s)}_s (x) = 0_s$.\qed
\end{proof}


\section{Commutativity and finiteness properties} \label{sec:finiteness}

The example of a network $f$ that satisfies \ref{item:Ci} but not \ref{item:Cs} used the fact that $f$, though it depended on all the variables $x_i$ ``globally'', did not depend on each $x_i$ ``individually'': changing only the value of any $x_i$ could not change the value of $f_j(x)$. We thus consider various finiteness properties for the local functions.

For any $\phi : Q^Z \to Q$ and any ordered pair $(x,y) \in Q^Z \times Q^Z$, we say $u \subseteq Z$ is an \textbf{influence} of $\phi$ for $(x,y)$ if 
\[
	\phi(x_{Z \setminus u}, y_u) = \phi(y), \text{ and } \phi( x_{Z \setminus t}, y_t ) \ne \phi(y) \ \forall t \subset u.
\]

\begin{lemma} \label{lem:influence}
Let $\phi : Q^Z \to Q$ and $x, y \in Q^Z$.
\begin{enumerate}[label=(\arabic*)]
\item \label{item:v}
If there exists $v \in \mathcal{FP}(Z)$ such that $\phi(x_{Z \setminus v}, y_v) = \phi(y)$, then there exists a finite influence of $\phi$ for $(x,y)$. 
\end{enumerate}
Let $u \subseteq Z$ be an influence of $\phi$ for $(x,y)$.
\begin{enumerate}[resume, label=(\arabic*)]
	\item \label{item:influence_delta} 
	$u \subseteq \Delta(x,y) := \{i \in Z : x_i \ne y_i \}$.

	\item \label{item:influence_empty}
	$u = \emptyset$ if and only if $\phi(x) = \phi(y)$.

	\item \label{item:influence_subset}
	For any $t \subseteq u$, there exists $z \in Q^Z$ such that $t$ is an influence of $\phi$ for $(z,y)$.
\end{enumerate}
\end{lemma}


We then consider three finiteness properties for a function $\phi : Q^Z \to Q$.
\begin{enumerate}[label=\textbf{(F\arabic*)}]
	\item \label{item:F1}
	For any $x,y \in Q^Z$, there exists an influence of $\phi$ for $(x,y)$.

	\item \label{item:F2}
	For any $x,y \in Q^Z$, there exists a finite influence of $\phi$ for $(x,y)$.

	\item \label{item:F3}
	There exists $b \in \mathcal{FP}(Z)$ such that for any $x,y \in Q^Z$, there exists an influence of $\phi$ for $(x,y)$ contained in $b$.
\end{enumerate}
By extension, we say a network $f : Q^Z \to Q^Z$ satisfies \ref{item:F1} (\ref{item:F2}, \ref{item:F3}, respectively) if for all $i \in Z$, $f_i$ satisfies \ref{item:F1} (\ref{item:F2}, \ref{item:F3}, respectively). Clearly, $\ref{item:F3} \implies \ref{item:F2} \implies \ref{item:F1}$.



%
%
%
%

In order to emphasize the role of different commutativity properties, we introduce the following notation for our results: ${\bf (C)} \vdash {\bf (Y)} \implies {\bf (Z)}$ means that the implication ${\bf (Y)} \implies {\bf (Z)}$ holds when we restrict ourselves to networks with Property ${\bf (C)}$. This notation will be combined with the other ones introduced so far.

\begin{theorem} \label{th:finiteness}
The commutativity and finiteness properties are related as follows.
\begin{align}
	\label{item:F1Cs}
	\ref{item:Cs} &\vdash \ref{item:F1} \notimplies \ref{item:F2}.\\
	\label{item:F2CsF3} 
	\ref{item:Cs} &\vdash \ref{item:F2} \notimplies \ref{item:F3}. \\
	\label{item:CiF1Cs}
	\ref{item:Ci} &\vdash \ref{item:F1} \notimplies \ref{item:Cs}.\\	
	\label{item:CiF2Cs}
	\ref{item:Ci} &\vdash \ref{item:F2} \implies \ref{item:Cs}.
\end{align}
\end{theorem}

\begin{proof}
For the first two items, we only need to exhibit counterexamples for $Q = \{0,1\}$ and $Z = \N$ of the form $f_1(x) = \phi(x)$ and $f_i(x) = x_i$ for all $i \ge 2$, which always verify \ref{item:Cs}.

\eqref{item:F1Cs}. Let $\phi(x) = \bigwedge_{i \in \N} x_i$. We verify that $\phi$ satisfies \ref{item:F1} but does not satisfy \ref{item:F2}. Indeed, if $x \ne 1_{\N}$, say $x = (0_s, 1_t)$ and $y = 1_\N$, then $s$ is an influence for $(x,y)$ and any $i \in s$ is an influence for $(y,x)$.

%
%

\eqref{item:F2CsF3}. For any nonzero $x \in \{0,1\}^\N$, let $a(x) = \min\{ j : x_j = 1 \}$, and let 
\[
	\phi(x) = \begin{cases}
	0 &\text{if } x = 0_\N,\\
	x_{a(x) + 1} &\text{otherwise}.
	\end{cases}
\]
We verify that $\phi$ satisfies \ref{item:F2}; we shall use Lemma \ref{lem:influence}\ref{item:v} repeatedly. If $y \ne 0_\N$, then for any $x$, $\phi( x_{\N \setminus [a(y) + 1]}, y_{[a(y) + 1]} ) = \phi(y)$. If $y = 0_\N$, then for any $x \ne 0_\N$, $\phi( x_{ \N \setminus \{ a(x) + 1 \} }, y_{\{ a(x) + 1 \}} ) = \phi(y)$. So $\phi$ satisfies \ref{item:F2} but the case $y = 0_\N$ shows that it does not satisfy \ref{item:F3}.

\eqref{item:CiF1Cs}. Let $Q = \{0,1\}$, $Z = \N$, and for all $i \in \N$,
\[
	f_i(x) = x_i \lor \bigvee_{k > i} \neg x_k.
\]
In other words, $f_i(x) = x_i$ if $x_k=1$ for all $k > i$ and $f_i(x) = 1$ whenever there exists $k > i$ with $x_k = 0$. It is clear that $f$ satisfies \ref{item:F1}. We now prove that $f$ satisfies \ref{item:Ci}. Let $i,j \in \N$ be distinct and $x \in \{0,1\}^\N$. Clearly, if $i > j$, then $f_i(x) = f_i( f^{(j)}(x) )$, so suppose $i < j$. If there exists $k > j$ with $x_k = 0$, then $f_i(x) = f_i( f^{(j)}(x) ) = 1$; otherwise, $f^{(j)}(x) = x$ and $f_i(x) = f_i( f^{(j)}(x) )$ and we are done. We finally prove that $f$ does not satisfy \ref{item:Cs}. Let $x = 0_\N$, then $f_1(x) = 1$, while $f_1( f^{(\N \setminus 1)} (x) ) = f_1( 0, 1_{\N \setminus 1} ) = 0$.

\eqref{item:CiF2Cs}. We first show that \ref{item:Cs} is equivalent to: $f^{(\sigma, \tau)} = f^{(\tau, \sigma)}$ for all $\sigma, \tau \in \mathcal{P}(Z)$ with $\sigma \cap \tau = \emptyset$. Indeed, in the latter case, for any $s, t \in \mathcal{P}(Z)$, we have ${f^{(s \cap t, s \setminus t)}=f^{(s \setminus t,s \cap t)}=f^{(s)}}$ and symmetrically for $f^{(t)}$ so that
\[
	f^{(s,t)} = f^{(s \cap t, s \setminus t, t \setminus s, s \cap t )} = f^{(s \cap t, t \setminus s, s \setminus t, s \cap t )} = f^{(t,s)}.
\]

Now, suppose $f$ satisfies \ref{item:F2} and \ref{item:Ci}, but not \ref{item:Cs}. Then there exist $s,t$ and $x$ for which $s \cap t = \emptyset$ and $f^{(s,t)}_t(x) \ne f^{(t,s)}_t(x) = f_t(x)$. In particular, there exists $i \in t$ such that $f_i( f^{(s)} (x) )\ne f_i(x)$. Denoting $y = f^{(s)}(x)$, there exists a finite influence $b \subseteq \Delta(x,y) \subseteq s$ such that $f^{(b)}(x) = (x_{Z \setminus b}, y_b)$ verifies $f_i( f^{(b)}(x) ) = f_i(y) \ne f_i(x)$. This implies $f^{(i,b)} \ne f^{(b,i)}$, which is the desired contradiction.\qed
\end{proof}

\section{Commutativity and dynamical locality} \label{sec:dynamic_locality}

Let $X$ be a set, and $\alpha : X \to X$. A \textbf{cycle} of $\alpha$ is a finite sequence $x^1, \dots, x^l \in X$ such that $\alpha(x^i) = x^{i+1}$ for $1 \le i \le l - 1$ and $\alpha(x^l) = x^1$. The integer $l$ is the length of the cycle; the \textbf{period} of $\alpha$ is the least common multiple of all the cycle lengths of $\alpha$. (If $\alpha$ has no cycles, or if it has cycles of unbounded lengths, then its period is infinite.) The transient length of $x$ is the smallest $k \ge 0$, such that $\alpha^k(x)$ belongs to a cycle of $\alpha$. The \textbf{transient length} of $\alpha$ is the maximum over all transient lengths. (Again, if $\alpha$ has no cycles, or if $\alpha$ has unbounded transient lengths, then the transient length is infinite.)

Here is an example of $\alpha$ where every trajectory leads to a cycle, but $\alpha$ has infinite period and infinite transient length. Let $X = \N$, then for any prime number $p$, let 
\[
	\alpha(p^i) = \begin{cases}
		p^{i+1} &\text{if } 0 \le i \le 2p-2,\\
		p^p &\text{if } i = 2p-1,
	\end{cases}
\]
and $\alpha(n) = n$ for any other $n \in \N$. Then the trajectory of the prime number $p$ has transient length $p$ and period $p$.

We note that for any $m > n \ge 0$, $\alpha^m = \alpha^n$ if and only if $\alpha$ has transient length $\le n$ and period dividing $m - n$. In particular, any $\alpha : Q \to Q$ has transient length at most $q-1$ and period at most $q$. Thus, for $\pi_q := \lcm(1,2, \dots, q)$, any $\alpha : Q \to Q$ satisfies
\begin{equation} \label{equation:alpha}
	\alpha^{\pi_q + q - 1} = \alpha^{q-1}.
\end{equation}
Moreover, this is the minimum equation satisfied by all $\alpha : Q \to Q$, in the sense that any equation of the form $\alpha^m = \alpha^n$ for $m > n \ge 0$ must have $m \ge \pi_q + q - 1$ and $n \ge q-1$.
We then consider the dynamical property
\begin{enumerate}[label=\textbf{(D)}]
	\item \label{item:D} 
	$f$ has transient length at most $q-1$ and period at most $q$, i.e. $f^{\pi_q + q - 1} = f^{q-1}$.
\end{enumerate}
We shall refer to Property \ref{item:D} as being \textbf{dynamically local}, as it implies that $f$ behaves like a function $Q \to Q$. We naturally also consider its analogues for updates.
\begin{enumerate}[label=\textbf{(D\arabic*)}]
	\item \label{item:Di} 
	$f^{(i)}$ is dynamically local for all $i \in Z$.
	
	\item \label{item:Db} 
	$f^{(b)}$ is dynamically local for all $b \in \mathcal{FP}(Z)$.
	
	\item \label{item:Ds} 
	$f^{(s)}$ is dynamically local for all $s \in \mathcal{P}(Z)$.
\end{enumerate}

Property \ref{item:Di} is actually trivial, as it is satisfied by any network. Indeed, $f^{(i)}$ can be decomposed into a family of mappings from $Q$ to itself (one for each value of $x_{Z \setminus i}$): for any $a \in Q^{Z \setminus i}$, let $g^a : Q \to Q$ be $g^a(x_i) := f_i(x_i, a)$, then
\[
	f^{(i)^m}(x_i,a) = ( \left( g^a \right)^m (x_i), a) 
\]
for all $m \ge 1$, thus $f^{(i)}$ verifies Equation \eqref{equation:alpha}. 

%
%
%

\begin{theorem} \label{th:CD}
Commutativity and dynamical locality are related as follows.
\begin{align}
	\label{item:CsDs}
	\ref{item:Cs} &\implies \ref{item:Ds}.\\
	\label{item:CiDb}
	\ref{item:Ci} &\implies \ref{item:Db}.\\
	\label{item:CiD}
	\ref{item:Ci} &\notimplies \ref{item:D}.\\
	\label{item:CiDs}
	\ref{item:Ci} &\vdash \ref{item:D} \notimplies \ref{item:Ds}.
\end{align}
\end{theorem}


If $f : Q^Z \to Q^Z$ is dynamically local, then the following are equivalent: $f$ is bijective; $f$ is injective; $f$ is surjective; $f^{\pi_q} = \id$. We thus consider the property
\begin{enumerate}[label=\textbf{(B)}]
	\item  \label{item:B}
	$f$ is bijective.
\end{enumerate}
Again, we consider its counterparts for updates.
\begin{enumerate}[label=\textbf{(B\arabic*)}]
	\item \label{item:Bi} 
	$f^{(i)}$ is bijective for all $i \in Z$.
	
	\item \label{item:Bb} 
	$f^{(b)}$ is bijective for all $b \in \mathcal{FP}(Z)$.
	
	\item \label{item:Bs} 
	$f^{(s)}$ is bijective for all $s \in \mathcal{P}(Z)$.
\end{enumerate}

By using dynamical locality and similar arguments to those used in the proof of Theorem \ref{th:CD}, we obtain that all versions of bijection are equivalent for globally commutative networks; however, locally commutative networks are not so well behaved.

%
%
%

\begin{theorem}
  \label{theo:CB}
  Bijection properties and commutativity properties are related as follows.
\begin{align}
	\label{item:CsB}
	\ref{item:Cs} &\vdash \ref{item:B} \iff \ref{item:Bi} \iff \ref{item:Bb} \iff \ref{item:Bs}.\\
	\label{item:CiB} 
	\ref{item:Ci} &\vdash \ref{item:Bs} \implies \ref{item:B} \implies \ref{item:Bi} \iff \ref{item:Bb}.\\
	\label{item:CinotBs}
	\ref{item:Ci} &\vdash \ref{item:B} \land \ref{item:Ds}  \notimplies \ref{item:Bs}. \\
	\label{item:CinotB}
	\ref{item:Ci} &\vdash \ref{item:Bi} \land \ref{item:Ds} \notimplies \ref{item:B}.
\end{align}
\end{theorem}


\section{Commutativity and idempotence} \label{sec:idempotence}

Let us now strengthen the commutativity properties as follows.
\begin{enumerate}[label=\textbf{(IC\arabic*)}]
	\item \label{item:ICi}
	$f^{(i,j)} = f^{( \{ i,j \} )}$ for all $i,j \in Z$.

	\item \label{item:ICb}
	$f^{(b,c)} = f^{( b \cup c )}$ for all $b,c \in \mathcal{FP}(Z)$.

	\item \label{item:ICs}
	$f^{(s,t)} = f^{( s \cup t )}$ for all $s,t \in \mathcal{P}(Z)$.
\end{enumerate}
Intuitively, \ref{item:ICi} means that updating $i$ and $j$ in series is equivalent to updating them in parallel; \ref{item:ICb} and \ref{item:ICs} then extend this property to updates of finite blocks and to any updates, respectively. This is closely related to \textbf{idempotence}:
\begin{enumerate}[label=\textbf{(I)}]
	\item \label{item:I} 
	$f^2 = f$.
\end{enumerate}
Dynamically, idempotence means that $Q^Z$ is partitioned into gardens of Eden of $f$ (configurations $y$ such that $f^{-1}(y) = \emptyset$) and fixed points of $f$ (configurations $z$ such that $f(z) = z$). Again, we consider the counterparts of idempotence to updates.
\begin{enumerate}[label=\textbf{(I\arabic*)}]
	\item \label{item:Ii} 
	$f^{(i)^2} = f^{(i)}$ for all $i \in Z$.
	
	\item \label{item:Ib} 
	$f^{(b)^2} = f^{(b)}$ for all $b \in \mathcal{FP}(Z)$.
	
	\item \label{item:Is} 
	$f^{(s)^2} = f^{(s)}$ for all $s \in \mathcal{P}(Z)$.
\end{enumerate}

For globally commutative networks, all four notions of idempotence are equivalent. This is far to be the case for locally commutative networks instead.

%
%
%
%
%
%

\begin{theorem} \label{th:IC}
Idempotence properties and commutativity properties are related as follows.
\begin{align}
	\label{item:ICsCs}
	\ref{item:Cs} &\vdash  \ref{item:I} \iff \ref{item:Ii} \iff \ref{item:Is} \iff \ref{item:ICs}.\\
	\label{item:ICiICb}
	\ref{item:Ci} &\vdash \ref{item:Ii} \iff \ref{item:ICi} \iff \ref{item:Ib} \iff \ref{item:ICb}.\\
	\label{item:ICiD}
	\ref{item:Ci} &\vdash \ref{item:Ii}  \notimplies \ref{item:D}.\\
	\label{item:CiIDsIi}
	\ref{item:Ci} &\vdash \ref{item:I} \land \ref{item:Ds} \notimplies \ref{item:Ii}.\\
	\label{item:ICiI}
	\ref{item:Ci} &\vdash \ref{item:Ii} \land \ref{item:I} \notimplies \ref{item:Ds}.\\
	\label{item:CiIiDs}
	\ref{item:Ci} &\vdash \ref{item:Ii} \land \ref{item:Ds} \notimplies \ref{item:I}.\\
	\label{item:CiIsCs}
	\ref{item:Ci} &\vdash \ref{item:Is} \notimplies \ref{item:Cs}.
\end{align}
\end{theorem}

\begin{proof}
\eqref{item:ICsCs}. Clearly, $\ref{item:ICs} \implies \ref{item:Cs} \land \ref{item:Is}$. We first prove that $\ref{item:Cs} \land \ref{item:Is} \implies \ref{item:ICs}$. For any $s,t \in \mathcal{P}(Z)$, let $u = s \cap t$, then we have
\[
	f^{(s,t)} = f^{(s \setminus t, u, u, t \setminus s)} = f^{(s \setminus t, u, t \setminus s)} = f^{(s \cup t)}.
\]
We now prove that $\ref{item:Cs} \land \ref{item:I} \implies \ref{item:Is}$: for any $s \in \mathcal{P}(Z)$, we have
\[
	f^{(s)^2}_s = f^{(s, s, Z\setminus s, Z \setminus s)}_s = f^2_s = f_s,
\]
and hence $f^{(s)^2} = f^{(s)}$. We finally prove that $\ref{item:Cs} \land \ref{item:Ii} \implies \ref{item:Is}$: for any $s \in \mathcal{P}(Z)$ and any $i \in s$, we have
	${f^{(s)^2}_i = f^{ (i)^2 (s \setminus i)^2 }_i = f^{(i)^2}_i = f^{(i)}_i = f_i}$,
and hence $f^{(s)^2} = f^{(s)}$.

\eqref{item:ICiICb}. Clearly, $\ref{item:ICb} \implies \ref{item:ICi} \implies \ref{item:Ci} \land \ref{item:Ii}$ on the one hand and $\ref{item:ICb} \implies \ref{item:Cb} \land \ref{item:Ib}$ on the other. We now prove $\ref{item:Ci} \land \ref{item:Ii} \implies \ref{item:ICb}$. Let $f$ satisfy $\ref{item:Ci} \land \ref{item:Ii}$, and let $b, c \in \mathcal{FP}(Z)$. We denote $b \setminus c = \{ b_1, \dots,  b_k \}$, $c \setminus b = \{ c_1, \dots, c_l \}$ and $b \cap c = \{ d_1, \dots d_m \}$. Then, by \ref{item:Cbb},
\begin{align*}
	f^{(b,c)} &= f^{(b_1, \dots, b_k, d_1, \dots, d_m, c_1 \dots, c_l, d_1, \dots, d_m)}\\
	& = f^{(b_1, \dots, b_k, c_1 \dots, c_l, d_1, d_1, \dots, d_m, d_m)} \\
	&= f^{(b_1, \dots, b_k, c_1 \dots, c_l, d_1, \dots, d_m)} \\
	&= f^{(b \cup c)}.
\end{align*}

\medskip

For the remaining claims, we let $Q = \{0,1\}$ and $Z = \N$.

\eqref{item:ICiD}. Split $Z$ into parts $\{ Z_\omega : \omega \in \N \}$ of densities $2^{-\omega}$, and for any $\omega \in \N$, let 
\[
	f_{Z_\omega}(x) = \begin{cases}
		1_{Z_\omega} &\text{if } \delta(x) \ge 1 - 2^{1 -\omega},\\
		0_{Z_\omega} &\text{otherwise}.
	\end{cases}
\]
It is easy to verify that $f$ satisfies \ref{item:Ci} and \ref{item:Ii}. However, the initial configuration $x = 0_Z$ has an infinite trajectory, hence $f$ does not satisfy \ref{item:D}.

\eqref{item:CiIDsIi}. Let
\[
	f_1(x) = \begin{cases}
	\neg x_1 &\text{if } \delta(x) = 0,\\
	x_1 & \text{otherwise}.
	\end{cases}
\]
and $f_i(x) = 1$ for all $i \ge 2$. It is easy to verify that $f$ satisfies \ref{item:Ci}, \ref{item:I}, and \ref{item:Ds}. However, $f^{(1)}(0_\N) = (1, 0_{\N \setminus 1})$ and $f^{(1)^2}(0_\N) = 0_\N$, hence $f$ does not satisfy \ref{item:Ii}.

\eqref{item:ICiI}. Split $Z$ as above, and this time $f_{Z_1} = 0_{Z_1}$ and for any $\omega \ge 2$,
\[
	f_{Z_\omega}(x) = \begin{cases}
		1_{Z_\omega} &\text{if } \delta(x) > 1 - 2^{1 -\omega},\\
		x_{Z_\omega} &\text{otherwise}.
	\end{cases}
\]
Again, it is easy to verify that $f$ satisfies \ref{item:Ci} and \ref{item:Ii}. Moreover, $\delta( f(x) ) \le 1/2$ for any $x \in Q^Z$, hence $f(x)$ is fixed, thus $f$ satisfies \ref{item:I}. On the other hand, if $s = Z \setminus Z_1$, then $f^{(s)}$ has infinite trajectory for $x$ such that $x_{Z_1} = 1_{Z_1}$ and $\delta( x_{Z_\omega} ) = 2^{-\omega - 1}$ for all $\omega \ge 2$.

\eqref{item:CiIiDs}. Let 
\[
	f(x) = \begin{cases}
		1_\N & \text{if } \delta(x) = 0,\\
		0_\N & \text{if } \delta(x) = 1,\\
		x &\text{otherwise}.
	\end{cases}
\]
Then it is clear that $f$ satisfies \ref{item:Ci}, \ref{item:Ii} and \ref{item:Ds}; on the other hand, $f(0_\N) = 1_\N$ and $f^2(0_\N) = 0_\N$, hence $f$ is not idempotent.

\eqref{item:CiIsCs}. Let
\begin{align*}
	f_1(x) &= \begin{cases}
	x_1 &\text{if } \delta(x) = 1,\\
	0 &\text{otherwise}.
	\end{cases}\\
	f_i(x) &= 1 \qquad \forall i \ge 2.
\end{align*}
Then it is easy to verify that $f$ satisfies \ref{item:Ci} and \ref{item:Is}, but not \ref{item:Cs}.\qed
\end{proof}

\section{Globally commutative Boolean networks} \label{sec:classification}

There are a plethora of globally commutative networks over non-Boolean alphabets. For instance, for $q=4$, consider the following construction. Let $f$ be a Boolean network, and view $Q = \{0,1\}^2 = \{a = (a^1, a^2) : a^1, a^2 \in \{0,1\} \}$, then the quaternary network $g$ given by $g(x^1, x^2) = (f(x^2), x^2 )$ satisfies \ref{item:Cs}. This can be easily generalised for any $q \ge 4$. For $q = 3$, let $f$ be any Boolean network such that $f_i$ does not depend on $x_i$ for any $i$. Then let $g$ be the ternary network defined by
\[
	g_i(x) = \begin{cases}
	2 &\text{if } x_i = 2\\
	f_i( \hat{x} ) &\text{otherwise,}
	\end{cases}
\]
where $\hat{x}_i = \lfloor x_i/2 \rfloor$ for all $i$. Then it is easy to check that $g$ satisfies \ref{item:Cs}. However, we can classify globally commutative Boolean networks (i.e. networks with $Q = \{0,1\}$ and that satisfy \ref{item:Cs}). Before we give our classification, we need the following concepts and notation. First, for any $x \in \{0,1\}^Z$ and any $s \subseteq Z$, we denote $\overline{x^s} = (\neg x_s, x_{Z \setminus s})$.

The \textbf{transition graph} of $f$ is the directed graph $\Gamma(f)$ with vertex set $\{0,1\}^Z$ and an arc for every pair $(x , f^{(s)}(x) )$ for any $x \in \{0,1\}^Z$ and $s \in \mathcal{P}(Z)$. We remark that $\Gamma(f)$ completely determines $f$.

A \textbf{subcube} of $\{0,1\}^Z$ is any set of the form $\X[s, \alpha] := \{ x \in \{0,1\}^Z, x_s = \alpha \}$ for some $s \subseteq Z$ and $\alpha \in \{0,1\}^s$. A family of subcubes $X = \{ X_\omega : \omega \in \Omega \}$ is called an \textbf{arrangement} if $X_\omega \cap X_\xi \ne \emptyset$ for all $\omega, \xi \in \Omega$ and $X_\omega \not\subseteq X_\xi$ for all $\omega \ne \xi$.

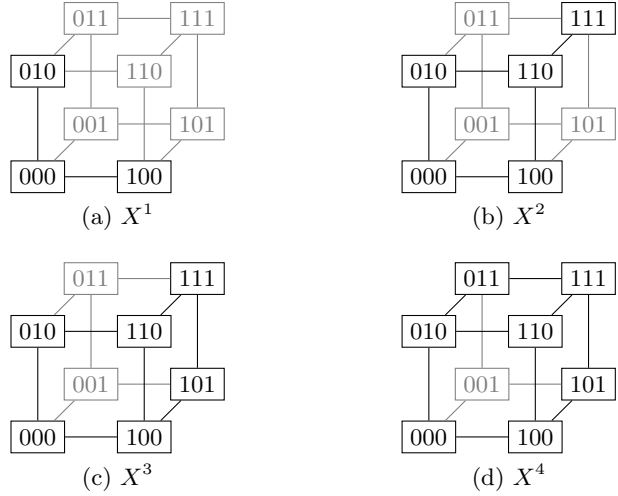
\begin{figure}
\centering

\subfloat[$X^1$]{
\begin{tikzpicture}[scale=.7]

\node[draw](000) at (0,0) {000};
\node[draw, color = gray](001) at (1,1) {001};
\node[draw](010) at (0,2) {010};
\node[draw, color = gray](011) at (1,3) {011};
\node[draw](100) at (2,0) {100};
\node[draw, color = gray](101) at (3,1) {101};
\node[draw, color = gray](110) at (2,2) {110};
\node[draw, color = gray](111) at (3,3) {111};

\draw[color=gray] (000) -- (001);
\draw (000) -- (010);
\draw (000) -- (100);

\draw[color=gray] (011) -- (010);
\draw[color=gray] (011) -- (001);
\draw[color=gray] (011) -- (111);

\draw[color=gray] (101) -- (100);
\draw[color=gray] (101) -- (111);
\draw[color=gray] (101) -- (001);

\draw[color=gray] (110) -- (111);
\draw[color=gray] (110) -- (100);
\draw[color=gray] (110) -- (010);
\end{tikzpicture}
} \hspace{2cm}
\subfloat[$X^2$]{
\begin{tikzpicture}[scale=.7]

\node[draw](000) at (0,0) {000};
\node[draw, color = gray](001) at (1,1) {001};
\node[draw](010) at (0,2) {010};
\node[draw, color = gray](011) at (1,3) {011};
\node[draw](100) at (2,0) {100};
\node[draw, color = gray](101) at (3,1) {101};
\node[draw](110) at (2,2) {110};
\node[draw](111) at (3,3) {111};

\draw[color=gray] (000) -- (001);
\draw (000) -- (010);
\draw (000) -- (100);

\draw[color=gray] (011) -- (010);
\draw[color=gray] (011) -- (001);
\draw[color=gray] (011) -- (111);

\draw[color=gray] (101) -- (100);
\draw[color=gray] (101) -- (111);
\draw[color=gray] (101) -- (001);

\draw (110) -- (111);
\draw (110) -- (100);
\draw (110) -- (010);
\end{tikzpicture}
}

\subfloat[$X^3$]{
\begin{tikzpicture}[scale=.7]

\node[draw](000) at (0,0) {000};
\node[draw, color = gray](001) at (1,1) {001};
\node[draw](010) at (0,2) {010};
\node[draw, color = gray](011) at (1,3) {011};
\node[draw](100) at (2,0) {100};
\node[draw](101) at (3,1) {101};
\node[draw](110) at (2,2) {110};
\node[draw](111) at (3,3) {111};

\draw[color=gray] (000) -- (001);
\draw (000) -- (010);
\draw (000) -- (100);

\draw[color=gray] (011) -- (010);
\draw[color=gray] (011) -- (001);
\draw[color=gray] (011) -- (111);

\draw (101) -- (100);
\draw (101) -- (111);
\draw[color=gray] (101) -- (001);

\draw (110) -- (111);
\draw (110) -- (100);
\draw (110) -- (010);

\end{tikzpicture}
}
\hspace{2cm}
\subfloat[$X^4$]{
\begin{tikzpicture}[scale=.7]

\node[draw](000) at (0,0) {000};
\node[draw, color = gray](001) at (1,1) {001};
\node[draw](010) at (0,2) {010};
\node[draw](011) at (1,3) {011};
\node[draw](100) at (2,0) {100};
\node[draw](101) at (3,1) {101};
\node[draw](110) at (2,2) {110};
\node[draw](111) at (3,3) {111};

\draw[color=gray] (000) -- (001);
\draw (000) -- (010);
\draw (000) -- (100);

\draw (011) -- (010);
\draw[color=gray] (011) -- (001);
\draw (011) -- (111);

\draw (101) -- (100);
\draw (101) -- (111);
\draw[color=gray] (101) -- (001);

\draw (110) -- (111);
\draw (110) -- (100);
\draw (110) -- (010);

\end{tikzpicture}
}
\caption{Arrangements} \label{fig:arrangements}
\end{figure}

We denote the \textbf{content} of $X$ by $\hat{X} := \bigcup_{\omega \in \Omega} X_\omega$.

\begin{lemma} \label{lem:intersection_arrangment}
Let $X = \{X_\omega : \omega \in \Omega \}$ be an arrangment, then $Y := \bigcap_{\omega \in \Omega} X_\omega$ is a non-empty subcube.
\end{lemma}

\begin{proof}
Denote $X_\omega := \X[s^\omega, \alpha^\omega]$ for all $\omega \in \Omega$. Then for any $\omega, \xi$, we have $\alpha^\omega_{s^\omega \cap s^\xi} = \alpha^\xi_{s^\omega \cap s^\xi}$, therefore we can define $\sigma = \bigcup_{\omega \in \Omega} s^\omega$ and $\alpha \in \{0,1\}^\sigma$ with $\alpha_{s^\omega} = \alpha^\omega_{s^\omega}$ for all $\omega \in \Omega$. Then it is easy to verify that $Y = \X[\sigma, \alpha]$.\qed
\end{proof}

For any $C \subseteq \{0,1\}^Z$, we classify any $i \in Z$ as follows.
\begin{itemize}
	\item If $x_i = y_i$ for any $x, y \in C$, then $i$ is an \textbf{external} dimension of $C$. Otherwise, $i$ is an \textbf{internal} dimension of $C$.
	
	\item If for any $x \in C$, $\overline{x^i} \in C$, then $i$ is a \textbf{free} dimension of $C$.
	
	\item If $i$ is an internal, non-free dimension of $C$, then $i$ is a \textbf{tight} dimension of $C$. If $i$ is a tight dimension, then there exists $z \notin C$ such that $\overline{z^i} \in C$; such $z$ is called an $i$-\textbf{border} of $C$.
\end{itemize}

If $X = \{X_\omega = \X[s^\omega, \alpha^\omega] : \omega \in \Omega\}$ is an arrangement, then (following the notation used in the proof of Lemma \ref{lem:intersection_arrangment}) the dimensions of $\hat{X}$ are as follows.
\begin{itemize}
	\item Let $\tau := \bigcap_{\omega \in \Omega} s^\omega$, then $\tau$ is the set of external dimensions of $\hat{X}$. The smallest cube containing $\hat{X}$ is $K(\hat{X}) := \X[ \tau, \alpha_\tau ]$.

	\item $\sigma := \bigcup_{\omega \in \Omega} s^\omega$ and $Z \setminus \sigma$ is the set of free dimensions of $\hat{X}$. The intersection subcube of $X$ is $Y:= \bigcap_{\omega \in \Omega} X_\omega = \X[\sigma, \alpha]$.
	
	\item The other dimensions in $\sigma \setminus \tau$ are the tight dimensions of $\hat{X}$.
\end{itemize}

For instance, let $Z = [3]$ and consider the following arrangements: 
\begin{align*}
	X^1 &= \{ (x_1, x_3) = (0,0) \} \cup \{ (x_2, x_3) = (0,0) \}\\
	X^2 &= \{ x_3 = 0 \} \cup \{ (x_1, x_2) = (1,1) \}\\
	X^3 &= \{ x_3 = 0 \} \cup \{ x_1 = 1 \}\\
	X^4 &= \{ x_3 = 0 \} \cup \{ x_1 = 1 \} \cup \{ x_2 = 1 \}.
\end{align*}
They are displayed in Figure \ref{fig:arrangements}. The dimensions of the different arrangements are classified as follows: for $X^1$, $1$ and $2$ are tight and $3$ is external; for $X^2$, all are tight; for $X^3$, $1$ and $3$ are tight and $2$ is free; for $X^3$, all are tight.

%
%
%
%
%
%
%

Let $C \subseteq \{0,1\}^Z$ and $f:\{0,1\}^Z \to \{0,1\}^Z$. For any $i \in Z$, we say $f_i$ is \textbf{trivial} on $C$ if 
$f_i(x) = x_i$ for all $x \in C$. We say $f_i$ is \textbf{uniform} on $C$ if for any $x, y \in C$, $x_i = y_i \implies f_i(x) = f_i(y)$. We say $f$ is \textbf{uniform nontrivial} on $C$ if $f_i$ is nontrivial and uniform on $C$ for all $i$.

We can then define a class of globally commutative Boolean networks by their transition graphs. Let $X$ be an arrangement. Outside of $\hat{X}$, $f$ is trivial: $f(x) = x$ if $x \notin \hat{X}$. In $\hat{X}$, $f$ satisfies the following:
\begin{enumerate}
	\item $f_i(x) = \alpha_i$ for every tight dimension $i$ of $\hat{X}$,

	\item $f_j$ is uniform nontrivial for any free dimension $j$ of $\hat{X}$,

	\item $f_k$ is trivial on any external dimension $k$ of $\hat{X}$.
\end{enumerate}
Any such network is referred to as an \textbf{arrangement network}.

For instance, the arrangement $X^3$ on Figure \ref{fig:arrangements} has three arrangement networks, one for each choice of the uniform nontrivial $f_2$ on $\hat{X}^3$. One such network, with $f_2(x) = \neg x_2$, is displayed in Figure \ref{fig:arrangement_network}.

\begin{figure}
\centering

\begin{tikzpicture}[scale=.7]

\node[draw](000) at (0,0) {000};
\node[draw, color = gray](001) at (1,1) {001};
\node[draw](010) at (0,2) {010};
\node[draw, color = gray](011) at (1,3) {011};
\node[draw](100) at (2,0) {100};
\node[draw](101) at (3,1) {101};
\node[draw](110) at (2,2) {110};
\node[draw](111) at (3,3) {111};

\draw[color=gray] (000) -- (001);
\draw[latex-latex] (000) -- (010);
\draw[-latex] (000) -- (100);

\draw[color=gray] (011) -- (010);
\draw[color=gray] (011) -- (001);
\draw[color=gray] (011) -- (111);

\draw[-latex] (101) -- (100);
\draw[latex-latex] (101) -- (111);
\draw[color=gray] (101) -- (001);

\draw[latex-] (110) -- (111);
\draw[latex-latex] (110) -- (100);
\draw[latex-] (110) -- (010);

\end{tikzpicture}
\caption{Arrangement network for $\hat{X}^3$} \label{fig:arrangement_network}

\end{figure}
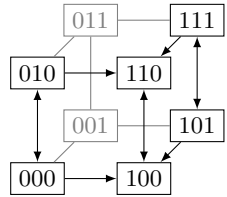

We can combine families of globally commutative networks as follows. For any $f$, a singleton $\{x\}$ is a connected component of $\Gamma(f)$ if and only if $x$ is an \textbf{unreachable fixed point} of $f$:
	${f^{(s)}(y) = x \iff y = x}$ for all ${s \in \mathcal{P}(X)}$.
Let $U(f)$ be the set of unreachable fixed points of $f$ and $R(f) = \{0,1\}^Z \setminus U(f)$. If $\{f^a : a \in A \}$ is a family of networks with $R(f^a) \cap R(f^{a'}) = \emptyset$ for all $a,a' \in A$ (or equivalently, $x \in U(f^a)$ or $x \in U(f^{a'})$ for any $x \in \{0,1\}^Z$), we define their \textbf{union} as
\[
	F(x) := \bigcup_{a \in A} f^a(x) = \begin{cases}
	f^a(x) 	&\text{if } x \in R(f^a)\\
	x 		&\text{otherwise.}
	\end{cases}
\]

It is easy to see that in an arrangement network $f$ for $X$, if $x \in X_\omega$ for some $\omega \in \Omega$, then so does $f^{(s)}(x)$ for any $s$. Therefore, the connected components of $\Gamma(f)$ are precisely $\hat{X}$ and all singletons $\{ \{y\} : y \notin \hat{X} \}$. It follows that if $f^a$ is an arrangement network for $X^a$, then the union $\bigcup_{a \in A} f^a$ is well defined if and only if $\hat{X^a} \cap \hat{X^{a'}} = \emptyset$ for all $a, a' \in A$.

It is then clear that $f$ is a union of arrangement networks if and only if, for every connected component $C$ of $\Gamma(f)$, the following holds:
\begin{enumerate}
	\item $C$ is the content of an arrangement,
	
	\item $f$ is uniform nontrivial on $C$,
	
	\item $f_i = cst = \neg z_i$ for any tight dimension $i$ and any $i$-border $z$ of $C$.
\end{enumerate}

\begin{theorem}
  \label{theo:boolean_cs}
  A Boolean network is globally commutative if and only if it is a union of arrangement networks.
\end{theorem}

If $X$ is an arrangement containing at least two subcubes, then $X$ has a tight variable, and hence no network for that arrangement is bijective. Conversely, if $\{X\}$ is an arrangement containing only one subcube, then there is only one bijective arrangement network for $X$, i.e. $f(x) = \neg x$ if $x \in X$ and $f(x) = x$ otherwise, which we shall refer to as the negation on $X$. We obtain the following classification of globally commutative, bijective Boolean networks.

\begin{corollary}
Let $f$ be a globally commutative, bijective Boolean network. Then $f$ is a union of negations on subcubes.
\end{corollary}

Therefore, the number $A(n)$ of globally commutative, bijective Boolean networks is equal to the number of partitions of the cube $\{0,1\}^n$ into subcubes. This, in turn, is equal to the number of minimally unsatisfiable cnfs on $n$ variables. The first few values of $A(n)$ are given in OEIS A018926.

\bibliographystyle{plain}
\bibliography{FDS}

\begin{thebibliography}{1}

\bibitem{BCG19}
Florian Bridoux, Alonso Castillo-Ramirez, and Maximilien Gadouleau.
\newblock Complete simulation of automata networks.
\newblock {\em Journal of Computer and System Sciences}, 109:1--21, May 2020.

\bibitem{CFG14}
Peter~J. Cameron, Ben Fairbairn, and Maximilien Gadouleau.
\newblock Computing in permutation groups without memory.
\newblock {\em Chicago Journal of Theoretical Computer Science},
  2014(07):1--20, 2014.

\bibitem{GN12}
E.~Goles and M.~Noual.
\newblock Disjunctive networks and update schedules.
\newblock {\em Advances in Applied Mathematics}, 48(5):646--662, 2012.

\bibitem{NS17}
Mathilde Noual and Sylvain Sen{\'e}.
\newblock Synchronism versus asynchronism in monotonic boolean automata
  networks.
\newblock {\em Natural Computing}, Jan 2017.

\end{thebibliography}

\newpage
\appendix

\section{Proof of Lemma~\ref{lem:influence}}

\begin{proof}
\ref{item:v}. Consider the family of sets $w \subseteq Z$ such that $\phi(x_{Z \setminus w}, y_w) = \phi(y)$. This family contains a finite set (i.e. $v$), thus it contains a set $u$ of minimum cardinality; it is clear that $u$ is then an influence for $(x,y)$.

\ref{item:influence_delta}. If $u \not\subseteq \Delta(x,y)$, then $t := u \cap \Delta(x,y)$ satisfies $t \subset u$ and $\phi( x_{Z \setminus t}, y_t ) = \phi( x_{Z \setminus u}, y_u ) = \phi(y)$, which is the desired contradiction.

\ref{item:influence_empty}. If $u = \emptyset$, then $\phi(x) = \phi(x_{Z \setminus u}, y_u) = \phi(y)$. Conversely, if $\phi(x) = \phi(y)$, then $\phi(x_{Z \setminus \emptyset}, y_\emptyset) = \phi(y)$ and there is no $t \subset \emptyset$, thus the empty set is an influence of $\phi$ for $(x,y)$.

\ref{item:influence_subset}. Let $v = u \setminus t$ and $z := (x_{Z \setminus v}, y_v)$, then $(x_{Z \setminus u}, y_u) = (z_{Z \setminus t}, y_t)$, from which we easily obtain that $t$ is an influence of $\phi$ for $(z,y)$. \qed
\end{proof}

\section{Proof of Theorem~\ref{th:CD}}

\begin{proof}
\eqref{item:CsDs}. If $f$ is globally commutative, then so is $f^{(s)}$ for any $s$; therefore, we only need to prove that \ref{item:Cs} $\implies$ \ref{item:D}. Now, let $f$ satisfy \ref{item:Cs}, $i \in Z$, then for any $m \ge 1$,
\[
	f^m_i 	= \left( f^{(i, Z \setminus i)^m} \right)_i = \left( f^{(i)^m (Z \setminus i)^m} \right)_i = f^{(i)^m}_i.
\]
Thus,
\[
	f_i^{\pi_q + q - 1} =  f^{ (i)^{\pi_q + q - 1} }_i = f^{ (i)^{q - 1} }_i = f^{q-1}_i.
\]
Since this holds for any $i$, we obtain $f^{\pi_q + q - 1} = f^{q-1}$.

\eqref{item:CiDb}. Let $f$ satisfy \ref{item:Ci} and let $b = \{b_1, \dots, b_k\} \in \mathcal{FP}(Z)$. Then for $m = \pi_q + q - 1$ and $n = q - 1$, we have
\[
	f^{(b)^m} = f^{(b_1, \dots, b_k)^m} = f^{((b_1)^m, \dots, (b_k)^m)} = f^{((b_1)^n, \dots, (b_k)^n)} = f^{(b_1, \dots, b_k)^n} = f^{(b)^n}.
\]

We delay the proofs of \eqref{item:CiD} and \eqref{item:CiDs} until Theorem \ref{th:IC}, where we prove stronger statements.\qed
\end{proof}

\section{Proof of Theorem~\ref{theo:CB}}

\begin{proof}
\eqref{item:CsB}. We only need to prove $\ref{item:Cs} \vdash \ref{item:Bi} \implies \ref{item:B} \implies \ref{item:Bs}$. We first prove $\ref{item:Cs} \vdash \ref{item:Bi} \implies \ref{item:B}$. In that case, we have for all $i \in Z$,
\[
	f_i^{\pi_q} = f_i^{ (i)^{\pi_q} (Z \setminus i)^{\pi_q} } = f_i^{ (Z \setminus i)^{\pi_q} } = \id_i,
\]
and hence $f^{\pi_q} = \id$. We now prove $\ref{item:Cs} \vdash \ref{item:B} \implies \ref{item:Bs}$. In that case, we have for all $s \in \mathcal{P}(Z)$,
\[
	f_s^{ (s)^{\pi_q} } = f_s^{ (s)^{\pi_q} (Z \setminus s)^{\pi_q} } = f_s^{\pi_q} = \id_s,
\]
and hence $f^{(s)^{\pi_q}} = \id$.

\eqref{item:CiB}. We first prove $\ref{item:Ci} \vdash \ref{item:Bi} \implies \ref{item:Bb}$. If $f$ satisfies \ref{item:Ci} and \ref{item:Bi}, then for any $b = \{b_1, \dots, b_k\} \in \mathcal{FP}(Z)$,
\[
	f^{(b)^{\pi_q}} = f^{ (b_1)^{\pi_q}, \dots, (b_k)^{\pi_q} } = \id.
\]
We now prove $\ref{item:Ci} \vdash \ref{item:B} \implies \ref{item:Bi}$. Suppose $f$ satisfies \ref{item:Ci} but not \ref{item:Bi}, then let $f^{(i)}$ be non-injective and $x, y \in Q^Z$ such that $f^{(i)}(x) = f^{(i)}(y)$. Then for any $j \in Z$,
\[
	f_j(x) = f_j( f^{(i)}(x) ) = f_j( f^{(i)}(y) ) = f_j(y),
\]
and hence $f(x) = f(y)$. Thus $f$ is not bijective.

\eqref{item:CinotBs}. Let $Q = \{0,1\}$, $Z = \N$ and
\[
	f(x) = \begin{cases}
	\neg x &\text{if } \delta(x) = 0 \text{ or } \delta(x) = 1,\\
	x & \text{otherwise}.
	\end{cases}
\]
It is easy to verify that $f$ satisfies \ref{item:Ci}, \ref{item:B} and \ref{item:Ds}. However, let $s \subseteq Z$ such that $y = (1_s, 0_{Z \setminus s})$ has density $1/2$, then $f^{(s)}(0_\N) = y = f^{(s)}(y)$ and hence $f^{(s)}$ is not bijective.

\eqref{item:CinotB}. Let $Q = \{0,1\}$, $Z = \N$ and
\[
	f(x) = \begin{cases}
	\neg x &\text{if } \delta(x) = 0,\\
	x & \text{otherwise}.
	\end{cases}
\]
It is easy to verify that $f$ satisfies \ref{item:Ci}, \ref{item:Bi} and \ref{item:Ds}, but $f(0_\N) = 1_\N = f(1_\N)$ and hence $f$ is not bijective.\qed
\end{proof}

\section{Proof of Theorem~\ref{theo:boolean_cs}}

\begin{lemma} \label{item:union_Cs}
Any union of arrangement networks satisfies \ref{item:Cs}.
\end{lemma}

\begin{proof}
The proof is the conjunction of the following two claims.

\begin{myclaim}
Any arrangement network is globally commutative. 
\end{myclaim}

\begin{proof}
Let $f$ be an arrangement network for $X$ and $s', t' \in \mathcal{P}(Z)$.  If $x \notin \hat{X}$, then $x$ is fixed, and hence $f^{(s', t')}(x) = f^{(t', s')}(x) = x$. Therefore, let us assume $x \in \hat{X}$. We have $f_\sigma = cst = \alpha_\sigma$, thus we consider $s = (Z \setminus \sigma) \cap s'$ and $t = (Z \setminus \sigma) \cap t'$. Because $f$ is uniform, we can express 
\[
	f_{s \cap t}(x) = \lambda(x_{s \cap t}), \qquad f_{s \setminus t}(x) = \mu(x_{s \setminus t}), \qquad f_{t \setminus s}(x) = \nu(x_{t \setminus s}),
\]
for $\lambda : \{0,1\}^{s \cap t} \to \{0,1\}^{s \cap t}$, $\mu : \{0,1\}^{s \setminus t} \to \{0,1\}^{s \setminus t}$, and $\nu : \{0,1\}^{t \setminus s} \to \{0,1\}^{t \setminus s}$. We obtain
\[
	f^{(s,t)}(x) = ( \alpha_\sigma, \lambda^2(x_{s \cap t}), \mu(x_{s \setminus t}), \nu(x_{t \setminus s}), x_{(Z \setminus \sigma) \setminus (s \cup t)} ) = f^{(t,s)}(x).\qed
\]
\end{proof}

%

\begin{myclaim}
If $\{f^a : a \in A \}$ is a family of globally commutative networks with $R(f^a) \cap R(f^{a'}) = \emptyset$ for all $a,a' \in A$, then their union is also globally commutative.
\end{myclaim}

\begin{proof}
Let $F = \bigcup_a f^a$, $x \in \{0,1\}^Z$ and $s,t \in \mathcal{P}(Z)$. If $x$ belongs to $R(f^a)$, then so do $F^{(s)}(x)$ and $F^{(t)}(x)$, and hence $F^{(s,t)}(x) = (f^a)^{(s,t)}(x) = (f^a)^{(t,s)}(x) = F^{(t,s)}(x)$. If $x$ does not belong to any $R(f^a)$, then by definition $x \in U(F)$ and $F^{(s,t)}(x) = F^{(t,s)}(x) = x$.\qed
\end{proof}
\qed
\end{proof}

\begin{lemma} \label{lemma:c3_is_arrangement}
If $f$ is globally commutative, then it is a union of arrangement networks.
\end{lemma}

\begin{proof}
In the sequel, $f$ satisfies \ref{item:Cs} and $C$ is a connected component of $\Gamma(f)$.

\begin{myclaim} \label{lemma:Delta}
If $f^w (x) = y$ for some finite word $w = (w_1, \dots, w_k)$ with $w_i \in \mathcal{P}(Z)$, then $f^{(\Delta(x,y))}(x) = y$.
\end{myclaim}

\begin{proof}
Let $\Delta := \Delta(x,y)$. Suppose $f^{(\Delta)}(x) \ne y$, then there exists $i \in \Delta$ such that $f_i(x) = x_i$ or in other words, $f^{(i)}(x) = x$. But then, for any word $w = (w_1, \dots, w_k)$ where $i$ appears exactly $l$ times, we have
\[
	f_i^{(w_1, \dots, w_k)}(x) = f_i^{(i)^l(w_1 \setminus i, \dots, w_k \setminus i)}(x) = f_i^{(w_1 \setminus i, \dots, w_k \setminus i)}(x) = x_i,
\]
and hence $f^w(x) \ne y$.\qed
\end{proof}

\begin{myclaim} \label{lemma:z}
If $x, y \in C$, then there exist $s, s' \in \mathcal{P}(Z)$ such that $f^{(s)}(x) = f^{(s')}(y)$.
\end{myclaim}

\begin{proof}
Since $x$ and $y$ belong to the same component, there exists a sequence of configurations $x = x^0, x^1, \dots, x^k = y \in \{0,1\}^Z$ and a corresponding sequence of subsets $s^0, \dots, s^{k-1} \in \mathcal{P}(Z)$ such that either $f^{(s^i)}(x^i) = x^{i+1}$ or vice versa, in alternation. Suppose $k$ is minimal. If $k = 1$, then $f^{(s^0)}(x) = f^{(\emptyset)}(y) = y$, and we are done. If $k = 2$, then either $f^{(s^0)}(x) = x^1 = f^{(s^1)}(y)$, or $x = f^{(s^0)}(x^1)$ and $y = f^{(s^1)}(x^1)$ in which case $f^{(s^1)}(x) = f^{(s^0)}(y)$. Now suppose $k \ge 3$. Without loss of generality, we have $f^{(s^0)}(x) = x^1 = f^{(s^1)}(x^2)$ and $f^{(s^2)}(x^2) = x^3$. But then $f^{(s^0, s^2)}(x) = f^{(s^1,s^2)}(x^2) = f^{(s^2)}(x^3) =: \tilde{x}$; denoting $\Delta := \Delta(x, \tilde{x})$, Claim \ref{lemma:Delta} shows that $f^{(\Delta)}(x) = \tilde{x}$. Thus, the sequence $x, \tilde{x}, x^3, \dots, x^k$ contradicts the minimality of $k$.\qed
\end{proof}

\begin{myclaim} \label{lemma:uniform}
$f$ is uniform nontrivial on $C$.
\end{myclaim}

\begin{proof}
Firstly, let $e$ be an external dimension of $C$, then $C \subseteq \{x_e = a\}$ for some $a \in \{0,1\}$ and hence $f_e = cst = a$. Secondly, let $i$ be an internal dimension of $C$. Then $f_i$ is nontrivial on $C$ for any $i$, for otherwise $C$ could be split into two components, one for each value of the coordinate $i$. Let $x \in C$ with $f_i(x) \ne x_i$, and suppose $y \in C$ has $x_i = y_i$. By Claim \ref{lemma:z}, there exist $s, s' \in \mathcal{P}(Z)$ such that $f^{(s)}(x) = f^{(s')}(y)$, and since ${x_i=y_i}$ we can suppose that ${i\not\in s\cup s'}$. Then 
\[
	f_i(y) = f^{(i, s')}_i(y) = f^{(s', i)}_i(y) = f^{(s,i)}_i(x) = f^{(i,s)}_i(x) = f_i(x),
\]
thus $f_i$ is uniform.\qed
\end{proof}

\begin{myclaim} \label{lemma:constant}
On $C$, $f_i = cst = \neg z_i$ for any tight dimension and any $i$-border $z$ of $C$.
\end{myclaim}

\begin{proof}
Suppose $f_i(x) = z_i$ for some $x \in C$ with $x_i \ne z_i$. Let $y = \overline{z^i} \in C$, we then have $f_i(y) = f_i(x) = z_i$ and hence $z \in C$. Thus, $f_i(x) = \neg z_i$ for all $x \in C$ with $x_i = \neg z_i$. Now, if $f_i(x) = z_i$ for some $x_i = z_i$, then $f_i$ is trivial, which contradicts Claim \ref{lemma:uniform}.\qed
\end{proof}

\begin{myclaim} \label{lemma:arrangement}
$C$ is the content of an arrangement.
\end{myclaim}

\begin{proof}
Suppose $C$ is not the content of an arrangement. Consider the decomposition of $C$ into maximal subcubes $C = \bigcup_{\omega \in \Omega} X_\omega$, then there exist $X_0, X_1$ and $i \in Z$ such that $X_0 \subseteq H_0 := \{ x : x_i = 0 \}$ and $X_1 \subseteq H_1 := \{ x : x_i = 1 \}$. We also consider the two subcubes $Y_0 = X_0 \cup \{\overline{x^i} : x \in X_0 \}$ and $Y_1 = X_1 \cup \{\overline{x^i} : x \in X_1 \}$. We note that $i$ is an internal dimension of $C$. If $i$ is a free dimension of $C$, then $Y_0 \subseteq C$, which contradicts the maximality of $X_0$. If $i$ is a tight dimension of $C$, then by Claim \ref{lemma:constant}, $f_i$ is constant on $C$. If $f_i = 1$, then $f^{(i)}(x) = \overline{x^i}$ for any $x \in X_0$, and again $Y_0 \subseteq C$; if $f_i = 0$ we similarly obtain $Y_1 \subseteq C$.\qed
\end{proof}

Then Claims \ref{lemma:uniform}, \ref{lemma:constant} and \ref{lemma:arrangement} conclude the proof.\qed
\end{proof}

\end{document}